%% file: root.tex
\documentclass[letterpaper, 10 pt, conference]{ieeeconf}  

\IEEEoverridecommandlockouts                              

\usepackage{graphics} 
\usepackage{epsfig} 

\usepackage{amsmath} 
\usepackage{amssymb}  

\usepackage{amsthm}
\usepackage{tikz}
\usepackage{flushend}
\usepackage[noadjust]{cite} 

\newtheorem{thm}{Theorem}
\newtheorem{defn}{Definition}
\newtheorem{rem}{Remark}

\title{\LARGE \bf
An Exact Characterisation of Flexibility in\\ Populations of Electric Vehicles
}

\author{Karan Mukhi and Alessandro Abate
\thanks{This work was supported by Réseau de Transport d'Électricité.}
\thanks{Karan Mukhi and Alessandro Abate are with the Department of Computer Science, University of Oxford, UK, OX1 3QG
        {\tt\small \{karan.mukhi,alessandro.abate\}@cs.ox.ac.uk}}%
}

\begin{document}

\maketitle
\thispagestyle{empty}
\pagestyle{empty}

\begin{abstract}
Increasing penetrations of electric vehicles (EVs) presents a large source of flexibility, which can be used to assist balancing the power grid. The flexibility of an individual EV can be quantified as a convex polytope and the flexibility of a population of EVs is the Minkowski sum of these polytopes. In general computing the exact Minkowski sum is intractable. However, exploiting symmetry in a restricted but significant case, enables an efficient computation of the aggregate flexibility. This results in a polytope with exponentially many vertices and facets with respect to the time horizon. We show how to use a lifting procedure to provide a representation of this polytope with a reduced number of facets, which makes optimising over more tractable. Finally, a disaggregation procedure that takes an aggregate signal and computes dispatch instructions for each EV in the population is presented. The complexity of the algorithms presented is independent of the size of the population and polynomial in the length of the time horizon. We evaluate this work against existing methods in the literature, and show how this method  guarantees optimality with lower computational burden than existing methods. 
\end{abstract}

\section{INTRODUCTION}

In the pursuit of a more sustainable energy system, increasing amounts of renewable energy sources are being integrated onto electric grids. Many of these sources of energy are intermittent in nature, causing generation to be variable and uncertain. 
Due to this, sources of flexibility are required to maintain an economic and reliable grid. Concurrently, mass adoption of electric vehicles (EVs) is set to occur within the next decade. Analysis of EV charging patterns show they are typically plugged in for more time than they require to charge \cite{Lee2019ACN-Data:Dataset}, and so there is a set of charging profiles they may take whilst satisfying their energy requirements.
Hence, by controlling their charging profiles, populations of EVs present an opportunity to procure vast amounts of flexibility and aid in balancing the grid.

Current systems for load control entail system operators issuing commands directly to large loads. 
However, this is not practicable when controlling the load of millions of distributed devices.
To enable scalability, hierarchical load control architectures have been proposed, whereby aggregators collate the flexibility from individual loads and present them to the system operator \cite{Callaway2011AchievingLoads}.
Aggregators maximise the utility of services they provide to the system operators, whilst satisfying the constraints of individual loads.
They are thus tasked with aggregating individual flexibility sets as efficiently as possible, as highlighted in 
\cite{EsmaeilZadehSoudjani2015AggregationAbstractions, Hao2015AggregateLoads,Hao2014CharacterizingLoads,Barot2017APolytopes, Trangbaek2011ExactControl,Zhao2017ALoads,Muller2019AggregationResources,Nazir2018InnerResources,Kundu2018ApproximatingApproach}.

The set of charging profiles of an EV that respect the energy requirements of the EV can be represented as a convex polytope. Geometrically, the aggregate flexibility of a population of EVs is the Minkowski sum of these polytopes. The Minkowski sum can be calculated by computing the convex hull of the sum of all sets of vertices of the summands \cite{Weibel2007MinkowskiComputation}. In general, this becomes infeasible when the dimensions of the polytopes, or the number of polytopes being summed, is large.
Consequently, recent work has focused on providing  approximations to the aggregate flexibility set. 
In \cite{Barot2017APolytopes} methods for computing outer approximations are presented, however, whilst guaranteeing the inclusion of optimal points, outer approximations also contain infeasible ones. This could cause penalties if an aggregator is contractually obliged to follow certain aggregate profiles.
The authors of \cite{Muller2019AggregationResources} use zonotopes to find inner approximations of the flexibility sets, allowing for an efficient computation of their Minkowski sum. Homothets also lead to an efficient Minkowski sum computation, and are used in \cite{Zhao2017ALoads} to find sufficient (maximum inner) and necessary conditions (minimum outer approximations) on the aggregate flexibility. This is extended in \cite{Nazir2018InnerResources} to provide a better approximation of the flexibility sets as the union of homothets. The Minkowski sum is viewed as a projection operator in \cite{Zhao2016ExtractingApproximation}, by finding the maximum homothet within the full dimensional polytope, they approximate this projection.
However, when optimising over inner approximations, by definition, optimality is not guaranteed. Furthermore the approximations may be so conservative that nominal charging profiles are not included in them, and so controlling the population may increase costs \cite{Ozturk2022AggregationAlgorithms}.



In this paper, we restrict ourselves to the case in which only charging is permitted. Because of this, we can leverage symmetries in the flexibility sets to provide an exact characterisation of the aggregate flexibility set. By considering the vertices of flexibility sets for EVs with homogeneous arrival and departure times we deduce that the sets are \textit{permutahedra}, a class of polytopes which permit efficient Minkowski addition. Permutahedra are closed under Minkowski summation, thus the aggregate flexibility set of these EVs is a permutahedra. The flexibility sets of EVs with heterogeneous arrival and departure times lie in different subspaces, hence aggregating them entails summing over permutahedra that belong to different subspaces.
By making arguments about which sets of vertices of the summands sum to give vertices of the summation,
we arrive at an exact characterisation of the aggregate flexibility set. The characterisation describes a polytope with exponentially many vertices and facets, in the length of the time horizon, making it intractable to optimise over. However, permutahedra permit an extended formulation that enables them to be optimised over in polynomial time \cite{Kaibel2011ExtendedOptimization}.
Therefore, we revert to the characterisation of the aggregate flexibility set as the Minkowski sum of permutahedra in various subspaces and reformulate the optimisation to a more tractable problem. Finally, we show how the solution to this problem can be used to derive a scheduling policy that disaggregates the optimal aggregate charging profile amongst the EVs in the population. In summary, the contributions of this paper are threefold. First, we describe an exact characterisation of the aggregate flexibility of populations of EVs. Second, we show how to efficiently optimise over this set. Finally, we provide a method of disaggregating an aggregate charging profile.

The rest of this paper is structured as follows: in Section \ref{sec:problem} we formalise the EV charging model and formulate the aggregation problem.
Section \ref{sec:agg} presents the exact characterisation of the aggregate flexibility set. 
In Section \ref{sec:opt} we propose a more tractable reformulation of the optimisation problem.
Section \ref{sec:disagg} describes how a scheduling policy can be derived from the solution of this problem. 
Case studies that benchmark the complexity and illustrate the utility of this work are provided in Section \ref{sec:case_studies}. 
Finally, we make conclusions in Section \ref{sec:conclusion}.

\section{PROBLEM FORMULATION}\label{sec:problem}
We consider an aggregator that has direct control over the charging profile of a population of $K$ EVs.
The \textit{charging profile}, $u(\tau)$, is defined as the power consumption of an EV at time $\tau \in \mathbb R^+$. 
We consider a finite and discrete time horizon comprising $n$ steps, each with duration $\delta$, 
and let $t \in \mathbb{T} := \{1,2,...,n\}$ index the time interval $[(t-1)\delta, t\delta)$. 
(In the following, we set $\delta = 1$ without loss of generality.) 
We thus consider the case where the charging profile $u(\tau)$ is a piecewise constant function 
\begin{equation*}
    u(\tau) = u_t \quad \forall \tau \in [(t-1), t) \quad \forall t \in \mathbb{T}.
\end{equation*} 
and can thus be embedded as a point $u \in \mathbb{R}^n$. 

The energy requirements of an EV are parameterised by the tuple $(E, a, d, m)$.  
The EV is defined as \textit{active} over the interval $\mathbb{C} := \{a, d-1\} \subseteq \mathbb{T}$, 
namely during each time step between its arrival, $a \in \{1,...,n\}$, and its departure, $d \in \{2,...,n+1\}$.  
Since $a<d$, let $p = d-a \leq n$ denote the number of time steps during which the EV is active.  
Whilst it is active, the EV is able to draw power, up to its power capacity $m$. 
Instead, when \textit{inactive} it does not consume any power. 
This imposes the following power constraints:
\begin{subequations}\label{eq:power_constraints}
    \begin{alignat}{2}
        &u_t = 0 \qquad &&\forall t \in \mathbb{T} \setminus \mathbb{C} , \label{constraint:def_loada}\\
        0 \leq &u_t \leq m \qquad &&\forall t \in\mathbb{C} . \label{constraint:def_loadb}
    \end{alignat}
\end{subequations}
The \textit{state of charge} (SoC) of the EV, $s_t$, is given by 
\begin{equation*}
    s_t = s_{t-1} + u_t,
\end{equation*}
where we assume $s_{a} = 0$ without loss of generality. 
As a requirement, the EV \textit{must} be fully charged at its departure time, $s_{d} = E$. 
Incorporating the power constraints from \eqref{eq:power_constraints}, 
this leads to the energy constraint 
\begin{equation}\label{eq:energy_constraints}
    \sum_{t\in \mathbb{C} } u_t = E. 
\end{equation}
As a result the set of feasible charging profiles that can be taken by an EV whilst respecting its energy requirements is given by
\begin{equation}\label{eq:resource_polytope}
    \mathcal{P} := \left\{ u \in \mathbb{R}^n \; \middle\vert \;
    \begin{array}{@{}cl}
        u_t = 0,  \quad &\forall t \in \mathbb{T}\setminus \mathbb{C};  \\
        0 \leq u_t \leq m,   \quad &\forall t \in \mathbb{C};  \\
       \sum_{t\in \mathbb{C} } u_t = E\\
    \end{array} 
    \right\},
\end{equation}
we will refer to this as the \textit{flexibility set} of an EV.
Note that the set $\mathcal{P}$ is non-empty for $E \leq (d-a)m$, whereas it is a singleton when the equality holds.

Let us now focus on a population of EVs. Let the superscript $i \in \mathbb{A} := \{1,...,K\}$ index the set of EVs, such that $u^i$ denotes the charging profile and $\mathcal{P}^i$ the flexibility set of the $i^{th}$ EV in the population.
The aggregator is tasked with finding a charging profile for each EV, such that the \textit{aggregate charging profile},
$x := \sum_i^K u^i$, minimises some cost, $f: \mathbb{R}^n \rightarrow \mathbb{R}$, over the time horizon. This is formulated:
\begin{equation}\label{eq:un_aggregated_opt}
    \begin{alignedat}{2}
    &\underset{u^i \; \forall i \in \mathbb{A}}{\text{minimise}} \;\; 
    &&f(x)\\
    &\text{subject to } \;\;            &&x = \sum_{i \in \mathbb{A}}u^i\\
    &                                   &&u^i \in \mathcal{P}^i.
    \end{alignedat}
\end{equation}

For large populations this problem can become intractable, therefore  reformulations of \eqref{eq:un_aggregated_opt} are required.
We define the \textit{aggregate flexibility set}, $\mathcal{P}^{agg}$, to be the set of feasible aggregate charging profiles that can be taken by the population as a whole:
\begin{equation*}
    \mathcal{P}^{agg} := \left\{ x = \sum^K_i u^i \middle|\; u^i \in \mathcal{P}^i, \;\; \forall \; i\right\}.
\end{equation*}
This set can be equivalently expressed as the Minkowski sum of the individual $\mathcal{P}^i$, namely:  
\begin{equation}\label{eq:minkowski_sum}
    \mathcal{P}^{agg} = \mathcal{P}^1 \oplus,..., \oplus \mathcal{P}^K = \biguplus_i^K \mathcal{P}^i. 
\end{equation}
In general computing Minkowski sums is NP-hard \cite{Tiwary2008OnPolytopes}. 
However, 
as we will show in the next section, symmetry in the polytopes $\mathcal{P}^i$ can be exploited to make the computation of such Minkowski sum be tractable. 

As anticipated earlier, within the constraints set by $\mathcal{P}^{agg}$, 
an aggregator is tasked with optimising some convex cost $f$, 
in order to find an optimal and feasible aggregate charging profile, call it $x^*$. 
Given such a global charging profile $x^*$ across the population, 
the aggregator should then find a set $\left\{  u^{1},..., u^{K} \right\}$ that disaggregates $x^*$ amongst the EVs in the population, 
so that $x^* = \sum_i^K u^{i}$, and $\forall \; i, u^{i} \in \mathcal{P}^i$. 


 

\section{AGGREGATION}\label{sec:agg}
To enable an efficient computation of the Minkowski sum it is convenient to first compute the Minkowski sum of EVs with homogeneous arrival and departure times, then evaluate the Minkowski sum on the resulting polytopes. 
We denote the subset of EVs with homogeneous arrival and departure times as 
\begin{equation*}
    \mathbb{A}^{a,d} := \left\{ i \in \{1,...,K\} \; \middle \vert \; a^i = a, \;\;d^i = d
    \right\}. 
\end{equation*}
For all $i \in \mathbb{A}^{a,d}$, equation \eqref{constraint:def_loada} constrains $\mathcal{P}^i$ to lie in the affine subspace $\mathbb{R}^p \subseteq \mathbb{R}^n$ . 
Let $\check{\mathcal{P}}^i \subset\mathbb{R}^p $ denote $\mathcal{P}^i$ in the subspace in which it is full dimensional.
Similarly we can start from a $\mathcal{Q}^{a,d} \subset \mathbb{R}^p$, and denote by $\hat{\mathcal{Q}}^{a,d}$ its embedding in $\mathbb{R}^n$. 

This allows us to characterise the Minkowski sum in \eqref{eq:minkowski_sum} more conveniently: 
namely,  
we first compute the Minkowski sum of flexibility sets that lie in the same subspace (cf. Sec. \ref{subsec:homog})
\begin{equation}    
    \mathcal{Q}^{a,d} := \biguplus_{i \in \mathbb{A}^{a,d}} \check{\mathcal{P}}^i,
\end{equation}
then sum the resulting polytopes (cf. Sec. \ref{subsec:heterog}) to obtain 
\begin{equation}\label{eq:hetero_q_sum}
\mathcal{P}^{agg} = \biguplus_{a,d} \hat{\mathcal{Q}}^{a,d}.
\end{equation}

\subsection{Homogeneous arrival and departure times}
\label{subsec:homog}
Consider an EV with charging requirements $(E,a,d,m)$ and flexibility set $\mathcal{P}$. Let $ q  := \left\lfloor  E/m \right\rfloor$ be the quotient and $r  :=  E - qm $ the modulus of $E$ with respect to $m$.
\vspace{.2mm}
\begin{defn} A \textit{permutahedron} $\Pi(x)$ is defined as the convex hull of all permutations of the vector $x \in \mathbb{R}^p$:
\begin{equation*}
    \Pi(x) := \textrm{ConvexHull}(x_\pi | \pi \in S_p), 
\end{equation*}
where $S_p$ is the symmetric group and $x_\pi$ denotes the permutation of the elements of $x$ specified by $\pi$.
Whenever we define $\Pi(x)$ we specify $x$ to be in the unique permutation in which it is \textit{monotone}, i.e. $x_1 \geq ... \geq x_p$. 
\end{defn}
\vspace{.2mm}
\begin{thm}\label{theorem:majorization_permutahedron}
$\check{\mathcal{P}}$ is the permutahedron $\Pi(v)$ where
\begin{equation}\label{eq:vertex_major}
    v  = ( \underbrace{m, ..., m}_{q}, r , \underbrace{0, ..., 0}_{p  - q  - 1}).
\end{equation}
\end{thm}

\begin{proof}
The constraint given by \eqref{constraint:def_loadb} restricts $\check{\mathcal{P}}$ to lie in a $p$ dimensional cube.
The edge length of this cube is $m$.
The constraint \eqref{eq:energy_constraints} defines a hyperplane.
The vertices of $\check{\mathcal{P}}$ lie at the intersection of the hyperplane described by \eqref{eq:energy_constraints} and the edges of the $p$-cube. It is clear to see that $v$ is a vertex of $\check{\mathcal{P}}$. It lies on the edge of the cube between vertices $(m, ..., m, 0 ,0, ..., 0)$ and $(m, ..., m, m , 0, ..., 0)$, and, using the definitions of $q$ and $r$, is in the hyperplane, $\sum^n_t v_t = E$. All other vertices are the permutations of the elements of $v$, this is shown in Fig. \ref{fig:cubePlane} for $p=3$ and $q=1$. By definition $r < m$, and so $v$ is the only monotone vertex of $\check{\mathcal{P}}$. We can then write $\check{\mathcal{P}}$ as $\check{\mathcal{P}} = \Pi(v)$.
\end{proof}

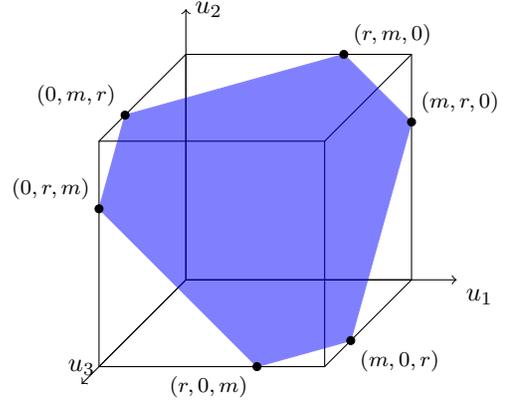
\begin{figure}[t]
  \centering
  \input{figures/cubePlane}
  \caption{Intersection of the plane $\sum_t^3 u_t = m + r$, and a cube with edge length $m$. Vertices are found by permuting the elements of  $v=(m,r,0)$.}
  \label{fig:cubePlane}
\end{figure}

\cite[Theorem 3]{Dahl2010Majorization01-matrices} states that, 
for permutahedra $\Pi(x)$ and $\Pi(y)$, 
their Minkowski sum is
\begin{equation}\label{eq:perm_m_sum}
    \Pi(x) \oplus \Pi(y) = \Pi(x+y). 
\end{equation}
Reintroducing the superscript $i$ indexing over EVs in the population, we define
\begin{equation}\label{eq:major_vertex_sum}
    \nu^{a,d} := \sum_{i \in \mathbb{A}^{a,d}} v^i.
\end{equation}
Theorem \ref{theorem:majorization_permutahedron} can be applied to \eqref{eq:perm_m_sum} in order to compute $\mathcal{Q}^{a,d}$: 
\begin{equation*}
    \mathcal{Q}^{a,d} = \biguplus_{i \in \mathbb{A}^{a,d}} \Pi(v^i) = \Pi(\nu^{a,d}).
\end{equation*} 
Note that a similar result is presented in \cite{Hao2014CharacterizingLoads} in the context of majorisation theory, 
however, we highlight that \cite{Hao2014CharacterizingLoads} leverages this only for the case where $m$ is homogeneous.

\subsection{Heterogeneous arrival and departure times} 
\label{subsec:heterog}
We now want to sum aggregate charging profiles for EVs with heterogeneous arrival and departure times, as described in \eqref{eq:hetero_q_sum}. 
As we have shown the previous section, this amounts to computing
\begin{equation}\label{eq:heterogeneous_sum}
    \mathcal{P}^{agg} = \biguplus_{a,d}\hat{\Pi}(\nu^{a,d}),
\end{equation}
namely the Minkowski sum of permutahedra that lie in different subspaces, all of which are embedded in the same ambient space $\mathbb R^n$. 

For $\pi \in S_n$, we write the permutation $\pi$ in its one-line notation e.g. an element of $S_6$ can be written as $\pi = (3,1,4,5,2,6)$. Furthermore, we define $\pi^{a,d} \in S_{d-a}$ as the order ranking of the elements of $\pi$ between the $a^{th}$ and $(d-1)^{th}$ elements, e.g. for $\pi = (3,1,4,5,2,6)$, we have 
\begin{alignat*}{2}
    &\pi^{2,6}  &&=(1,3,4,2), \\
    &\pi^{2,5}  &&=(1,2,3).
\end{alignat*}

Using this notation, we define 
\begin{equation}\label{eq:nu_hat_perm_notation}
    \hat{\nu}^{a,d}_\pi  := ( \underbrace{0, ..., 0}_{a-1}, \nu_{\pi^{a,d}}^{a,d} , \underbrace{0, ..., 0}_{n-d}).
\end{equation}

\begin{thm}
The vertices of $\mathcal{P}^{agg} = \biguplus_{a,d}\hat{\Pi}(\nu^{a,d})$ are given by the set $ \{\mu_\pi | \pi \in S_n\}$ where we define 
\begin{equation} 
    \mu_\pi := \sum_{a,d} \hat{\nu}^{a,d}_\pi.
\end{equation}
Hence, we can characterise the set of aggregate charging profiles as 
\begin{equation}\label{eq:complete_aggregation}
\mathcal{P}^{agg} =  \textrm{ConvexHull}\left( \mu_\pi \middle| \pi \in S_n \right).
\end{equation}
\end{thm}
The proof can be found in part \ref{appendix:heterogeneous_minkowski_sum_proof} of the Appendix. 

\begin{rem}
This is the only $\mathcal{V}$-representation of the exact aggregate flexibility set for a heterogeneous  population of EVs that the authors are aware of. The polytope in \eqref{eq:complete_aggregation} is an example of \textit{generalised permutahedron}. 
Recalling that $n$ denotes the finite number of discrete steps in the considered time horizon, this polytope in general has $n!$ vertices, and it can additionally be shown  that in general there are $2^n - 2$ facets \cite{Postnikov2009PermutohedraBeyond}. We also remark  that the number of facets and vertices is \textit{independent} of $K$, the size of the population, thus the characterisation is quite appropriate for arbitrarily large populations of EVs. 
\end{rem}

\section{OPTIMISATION}\label{sec:opt}
With a general characterisation of the aggregate flexibility set $\mathcal{P}^{agg}$ over the population of EVs, 
we are interested in computing an optimal aggregate charging profile, 
by solving the following constrained optimisation problem: 
\begin{equation}\label{eq:optimisation}
\begin{alignedat}{2}
&\underset{x}{\text{minimise}} \;\; &&f(x)\\
&\text{subject to } \;\; &&x \in \mathcal{P}^{agg}.
\end{alignedat}
\end{equation}
The objective $f$ can be quite general and is left momentarily undefined: we shall instead discuss specific instantiations in Sec. \ref{subsec:objectives}.  

\subsection{Tractable problem reformulation}\label{section:reformulation}
The characterisation of $\mathcal{P}^{agg}$ in \eqref{eq:complete_aggregation} was obtained to provide the aggregate flexibility set for a heterogeneous population of EVs.
As discussed, this formulation describes a polytope with in general exponentially many facets, as a function of the number of discrete steps $n$ over the time horizon.  
For large $n$, the optimisation over $\mathcal{P}^{agg}$ thus entails many constraints and is potentially  intractable. If instead we revert to the characterisation of $\mathcal{P}^{agg}$ given in \eqref{eq:heterogeneous_sum}, we can equivalently formulate the problem in a tractable way. 

\begin{defn}\label{def:birkhoff_polytope}
    The \textit{Birkhoff polytope}, denoted $B^p$, is the set of all $p \times p$ doubly stochastic matrices \cite{Ziegler2012LecturesPolytopes}, a complete description of which is given by  
    \begin{equation}\label{eq:birkhoff}
        B^p = \left\{ A \in \mathbb{R}^{p \times p} \; \middle\vert \; 
        \begin{array}{@{}cl}
            A_{ij} \geq 0     &\forall i, j \\
            \sum_i A_{ij} = 1 &\forall j\\
            \sum_j A_{ij} = 1 &\forall i\\
        \end{array} 
        \right\}.
    \end{equation}
\end{defn}
Extending the notation in the previous section for $B^p$ we define $\hat{B}^{a,d} \subset \mathbb{R}^{n\times n}$ as
\begin{equation*}
    \hat{B}^{a,d} := \left\{\;[0_{p,a-1}, A^T, 0_{p,n-d}]^T\;\middle|\; A \in  B^p \;\right\}, 
\end{equation*}
where $0_{m,n}$ denotes the $m\times n$ zero matrix. It can be shown that $\hat{\Pi}(\nu^{a,d})$ in  \eqref{eq:heterogeneous_sum} 
is the linear image of $\hat{B}^{a,d}$ under the mapping 
\begin{equation}\label{eq:birkhoff_perm_mapping}
    \hat{A}^{a,d} \mapsto \hat{A}^{a,d}\hat{\nu}^{a,d}, 
\end{equation}
whence we obtain an alternative description of $\mathcal{P}^{agg}$, as
\begin{equation*}
    \mathcal{P}^{agg} = \left\{ x = \sum_{a,d}\hat{A}^{a,d}\hat{\nu}^{a,d} \middle|
    \;\hat{A}^{a,d} \in \hat{B}^{a,d}, \;\; \forall\; a,d\right\},
\end{equation*}
which can be used to reformulate \eqref{eq:optimisation} as follows: 
\begin{equation}\label{eq:birkhoff_optimisation}
    \begin{alignedat}{2}
    &\underset{\hat{A}^{a,d} \; \forall\; a,d}{\text{minimise}} \;\; 
    &&f(x)\\
    &\text{subject to } \;\;            &&x = \sum_{a,d}\hat{A}^{a,d}\hat{\nu}^{a,d}\\
    &                               limit    &&\hat{A}^{a,d} \in \hat{B}^{a,d}, \;\; \forall\; a,d.
    \end{alignedat}
\end{equation}
The solution of this problem yields the optimal set $(\hat{A}^{a,d})^*$, for all $a,d$. 
From this the optimal aggregate charging profile $x^*$ can be recovered via the mapping $x^* = \sum_{a,d}(\hat{A}^{a,d})^*\nu^{a,d}$.
 
 \begin{rem}
     The solution to \eqref{eq:birkhoff_optimisation} is not unique.
      This is apparent from considering $\mathcal{P}^{agg}$ as the Minkowski sum from \eqref{eq:heterogeneous_sum}.
      Given a point, $z$, in the interior of $\mathcal{P}^{agg}$, one can consider the decomposition $z = \sum_{a,d} y^{a,d}$ where $y^{a,d} = \hat{A}^{a,d}\hat{\nu}^{a,d} \in \hat{\Pi}^{a,d}$. There is a set of such decompositions, and so we can conclude that the solution is not unique. Note that there is a unique decomposition iff $z$ is a vertex of $\mathcal{P}^{agg}$ \cite[Proposition 2.1]{Fukuda2004FromPolytopes}, from this we can deduce that the solution is unique iff $x^*$ is a vertex of $\mathcal{P}^{agg}$, this occurs when all $(\hat{A}^{a,d})^*$ are permutation matrices.

 \end{rem}

\subsection{Practical considerations on computational complexity}
The main motivation for this work is to compute  optimal charging profiles for EVs in a large and heterogeneous population.
If we were to optimise over the population without aggregating it first, there would be $kn$ decision variables, namely one for each EV at each time step therefore the complexity of this problem would be $O(nk)$. 
The optimisation problem formulated in \eqref{eq:birkhoff_optimisation} optimises over the set of matrices $\hat{A}^{a,d}$ for each pair of $a$ and $d$, with $a<d$. 
Each of the $\hat{A}^{a,d}$ introduce $(d-a)^2$ decision variables. 
Summing over all pairs of $a$ and $d$ we see that this problem has $n^2(n^2-1)/12$ decision variables, resulting in an overall complexity of $O(n^4)$, we emphasise that this is independent of the population size $k$. 

Aggregators operating on a national scale will generally deal with  populations of EVs on the order of a hundred thousand. Currently, the many energy markets are operated in half hourly settlement periods, when optimising over a 24hr time period this corresponds to $n = 48$. In this scenario aggregating charging profiles and optimising over this makes an otherwise intractable problem tractable. We do note however, that for smaller populations and over longer time periods, due to the quartic dependence of the length of the time horizon on complexity, this method of aggregation may actually increase complexity.

\subsection{Objective functions}\label{subsec:objectives}
As formulated in \eqref{eq:birkhoff_optimisation} any convex objective function, $f$, may be chosen. In scenarios where the price of energy, $p \in \mathbb{R}^n$, is fixed in each time step a linear objective, $f(x) = p^Tx$, can be used. 
Solutions to linear programs are vertices of the feasible the set, in this case $\mathcal{P}^{agg}$. These solutions may be obtained equivalently by solving linear programs for each EV in parallel and summing the individual optimal charging profiles \cite{Papadaskalopoulos2013DecentralizedMechanism}. This makes the process of computing $\mathcal{P}^{agg}$ redundant. However, we submit that formulations of the problem may include other, nonlinear terms in the objective, in which case optimal solutions will lie in the interior of $\mathcal{P}^{agg}$, and so characterising and optimising over $\mathcal{P}^{agg}$ is necessary. Examples of this occur when the price is a dependent on demand \cite{Gharesifard2016Price-basedResources}, when total demand from EV charging is non-negligible a fixed price is no longer a valid assumption. If the price of energy is affine in demand a quadratic cost, $f(x) = x^TPx + x^TPd$, may be used, where $P = diag(p)$ and $d \in \mathbb{R}^n$ represents the external demand. Solutions to the quadratic cost are not necessarily vertices of the feasible set, and so the strategy of solving in parallel is not applicable here.

Aside from minimising charging costs \eqref{eq:birkhoff_optimisation} can be formulated so that the aggregate charging profile tracks a generation signal. In this scenario the energy system operator specifies a generation signal $g \in \mathbb{R}^n$ that should be followed, the aggregator must then determine whethlimiter the generation signal is feasible, if it is not the aggregator may propose another generation signal that is feasible. Consequently the aggregator may specify the objective to be one that minimises some $p$-norm between the aggregate charging profile $x$ and $g$: $f(x) = ||x - g||_p $. If the optimal cost vanishes, the aggregator can deduce that $g$ is feasible. For non-zero optimal costs the aggregator can propose the optimal solution as a new generation signal that can be tracked, in which case the preferences of the system operator will guide the choice of norm.

\section{DISAGGREGATION}\label{sec:disagg}
Solving the problem in \eqref{eq:birkhoff_optimisation} will yield an optimal aggregate charging profile $x^*$.
However, this quantity gives no explicit indication of the individual charging profiles for each EV in the population. 
Therefore the aggregator ought to disaggregate the aggregate charging profile $x^*$ across the population, thus assigning a charging profile to each EV. We will refer to this assignment as a \textit{scheduler} and define it as 
\begin{equation*}
    \sigma(x^*) := \left[  u^{1^*},..., u^{K^*} \right], \\
\end{equation*}
where $u^{i^*} \in \mathbb{R}^n$ is the optimal charging profile assigned to EV $i$. To ensure the optimal aggregate charging profile is tracked, the scheduler must respect the constraint
 \begin{equation}\label{eq:disaggregation_criteria_tracks}
     x^* = \sum_i^K u^{i^*},
\end{equation}
moreover assignments must be feasible for each EV, and so
    \begin{equation}\label{eq:disaggregation_criteria_feasible}
       u^{i^*} \in \mathcal{P}^i, \;\;\forall \; i.
     \end{equation}

\begin{thm}
Let $(\hat{A}^{a,d})^*$ be the set of matrices that are the solution to \eqref{eq:birkhoff_optimisation}, and $v^i$ the only monotone vertex of $\check{\mathcal{P}}^i$, as defined in \eqref{eq:vertex_major}. The mapping
    \begin{equation*}
    u^{i^*} = (\hat{A}^{a,d})^* \hat{v}^i \quad \forall i \in \mathbb{A}^{a,d},
\end{equation*}
satisfies the criteria in \eqref{eq:disaggregation_criteria_tracks} and \eqref{eq:disaggregation_criteria_feasible} and is therefore a suitable scheduler. 
\end{thm}
The proof can be found in part \ref{appendix:disaggregation_proof} of the Appendix. This scheduler is dependent on $(\hat{A}^{a,d})^*$, which, as discussed in section \ref{section:reformulation}, is nonunique. Therefore, when $x^*$ is not a vertex of $\mathcal{P}^{agg}$, the disaggregation is also nonunique.

\section{CASE STUDIES}\label{sec:case_studies}
In this section we demonstrate the utility of the aggregation method presented in this paper. In Sec. \ref{subsec:benchmark} we use an energy arbitrage problem to benchmark time complexity. In Sec. \ref{subsec:tracking} we show how these methods can be used to track generation signals that are infeasible for other aggregation methods to track.

\subsection{Energy arbitrage}\label{subsec:benchmark} 
One of the functions of an aggregator is to minimise the cost of charging for EVs under their control.  We consider a scenario in limitwhich an aggregator bids into the day ahead market, where energy prices are fixed in each time step. The aggregator is tasked with minimising the aggregate charging cost of the population. Thus the problem is formulated as in \eqref{eq:birkhoff_optimisation} with the objective: $f(x) = p^Tx$. Note this choice of objective is used to reflect the fixed price energy arbitrage problem, however, these results hold for all convex objectives.
By virtue of these methods being an exact approximation, optimality is guaranteed, in contrast with the inner approximation methods in the literature. The only other method to guarantee optimality is to optimise over individual charging profiles i.e. solving the problem formulated in \eqref{eq:un_aggregated_opt}. Therefore we benchmark the time complexity against this method.

We optimise for varying population sizes, $K$, over time horizons, $n$, of different length, and plot the runtimes in Fig. \ref{fig:time_complexity_plot}.
For clarity, we have only plotted the runtime for optimising over the aggregated set for $K=8000$, as the optimisation problem is not dependent on the size of the population. Typical populations will orders of magnitude larger that this, and so from Fig. \ref{fig:time_complexity_plot} we see that the only tractable method for optimising the scheduling of large populations, over timesteps less than $60$, whilst guaranteeing optimality is by using the aggregation methods presented here. Furthermore, for the value of $n=48$, corresponding to the half hour day-ahead market, these methods can readily solve the problem.


\begin{figure}[t]
  \centering
  \includegraphics[scale=1.0]{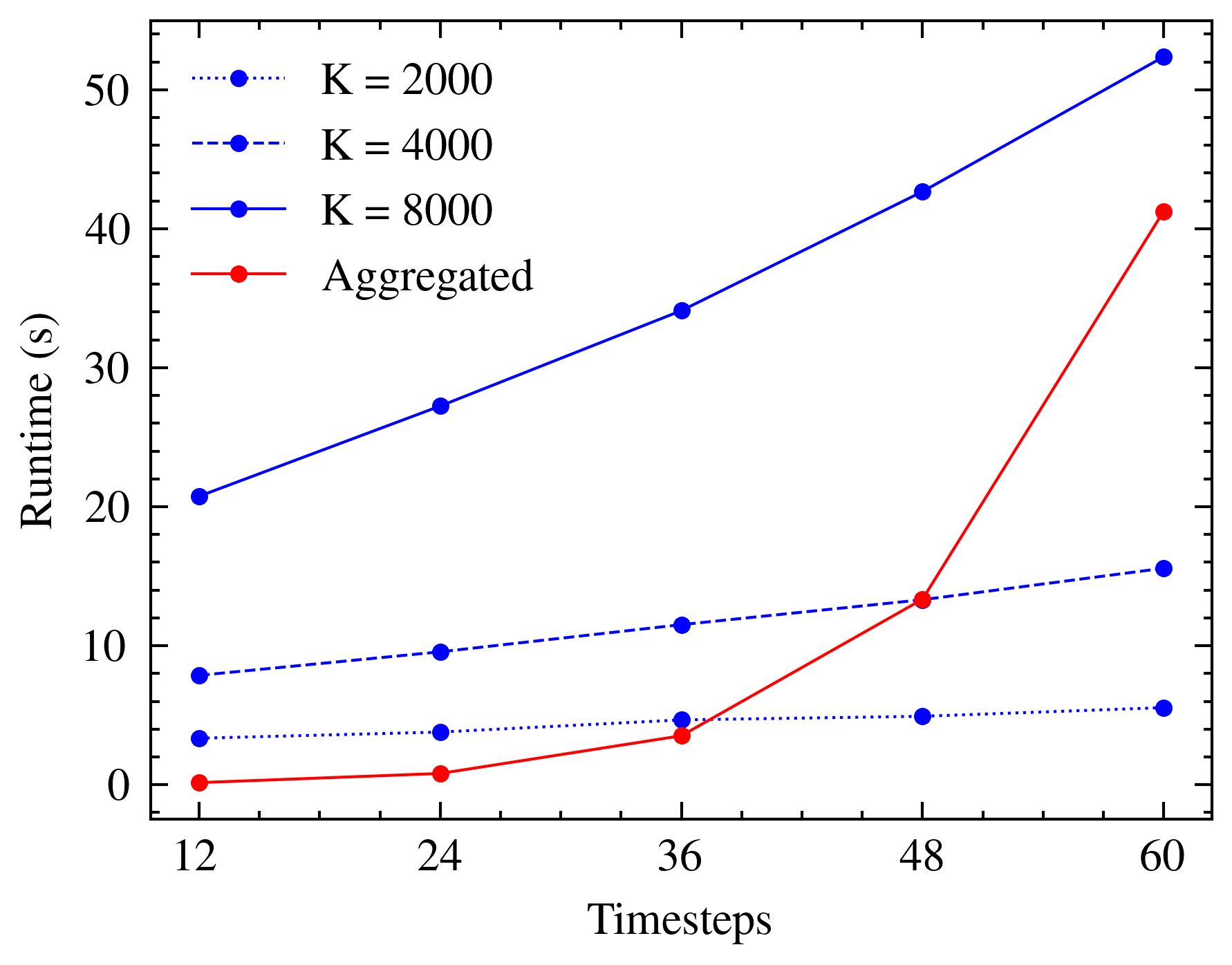}
  \caption{Runtimes for minimising the cost of charging for 
 a population of $K$ EVs. The runtime for solving the aggregated problem (red) is only plotted for $K=8000$, as this is constant in $K$.}
  \label{fig:time_complexity_plot}
\end{figure}

\subsection{Tracking generation signals}\label{subsec:tracking} 
In contrast with the exact characterisation provided in this paper, other methods in the literature have provided approximations of the aggregate flexibility set. 
To illustrate the benefit of an exact characterisation, we consider the task of tracking a generation signal. We simulate a population of $100$ EVs, with
arrival, and departure times sampled uniformly throughout a
24 hour time period. The power capacity, $m$, of each vehicle is sampled uniformly on the interval $[0.5,3]$kW and the required departure energy, $E$, is uniformly sampled from the set of values that permit the EV to be flexible. We benchmark our work with the zonotope \cite{Muller2019AggregationResources}, and homothet projection \cite{Zhao2016ExtractingApproximation}, approximation methods from the literature. Specifying $f$ to minimise the $l^2$ norm between the aggregate charging profile and a generation signal, we solve \eqref{eq:optimisation}. Fig. \ref{fig:reference_tracking_norms} shows the aggregate charging profile for each method along with the generation signal. We see aggregation method presented here is able to track the generation signal. Whilst the two approximation methods are unable to track the signal across the whole time horizon, showing that the (feasible) generation signal is not within their approximations.


\section{CONCLUSIONS AND FUTURE WORK}\label{sec:conclusion}
In this paper we have introduced an exact characterisation of the aggregate flexibility of a population of charging EVs. We have presented a  method to optimise over this set, with complexity that is constant in the number of EVs and $O(n^4)$ in the length of the time horizon, $n$. Finally, we have described how an aggregate charging profile, obtained as the solution of an optimisation problem, could be disaggregated across the population. Using an energy arbitrage problem we have shown that our method is tractable for values of $n\sim60$,
even when optimising for large populations,
and thus is applicable to scenarios corresponding to the day-ahead market. We have benchmarked our approach with methods in the literature, and shown that it provides optimal solutions (rather than sub-optimal).

Future work will generalise the approach to practically relevant cases, including allowing the state of charge at departure to lie within a range of values, and allowing discharging. It is also  of interest to study uncertain arrival and departure times, and to adapt the method to include network constraints. Finally, there may be a further reformulation of the optimisation problem that would admit a further reduction of complexity. This would enable these methods to be tractable when optimising in markets with longer time horizons.

\begin{figure}[t]
  \centering
  \includegraphics[scale=1.0]{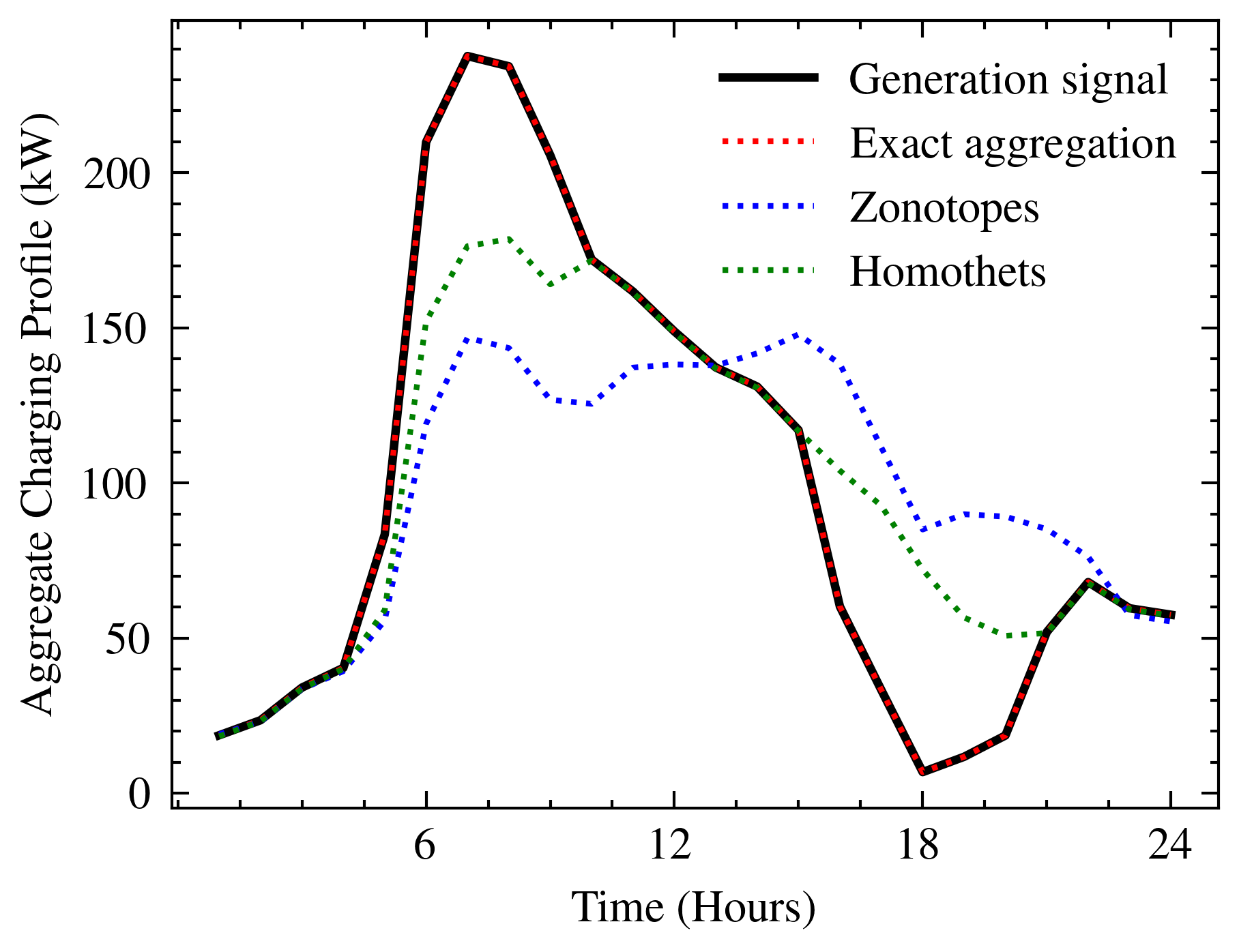}
  \caption{Aggregate charging profile when tracking a generation signal. The two approximation methods fail to track the generation signal over the time horizon.}
  \label{fig:reference_tracking_norms}
\end{figure}


\bibliographystyle{IEEEtran}
\bibliography{references}

\appendix
\subsection{Heterogeneous Minkowski Sum}\label{appendix:heterogeneous_minkowski_sum_proof}

\begin{proof}
Let $\mathcal{X} \subset \mathbb{R}^n$ be a convex polytope, the \textit{support function} of $\mathcal{X}$ is defined by
$h(\mathcal{X}, y) := \textrm{sup} \left\{ \langle x, y \rangle \middle| x \in \mathcal{X} \right\}$ for all $y \in \mathbb{R}^n$.
Support functions are Minkowski additive \cite{Weibel2007MinkowskiComputation}:
\begin{equation}\label{eq:support_function_minkowski_additive}
    h(\mathcal{X} \oplus \mathcal{X}', y) =  h(\mathcal{X}, y) +  h(\mathcal{X}', y).
\end{equation}
For strictly monotone $y$: $y_1 > ... > y_n$, the support of the permutahedron $\Pi(v)$ and the permutations of $y$: $\left\{y_\pi \middle| \pi \in S_n \right\}$, is given by the inner product of $y$ with the vertex of $\Pi$ in the corresponding permutation:
\begin{equation}
    h\left(\Pi(v), (y_\pi)\right) = \langle  v_\pi , y_\pi \rangle.
\end{equation}
All vertices of $\Pi(v)$ can be expressed in this way. Furthermore, when $\Pi \subset \mathbb{R}^p$, $p \leq n$, it is clear that

\begin{equation}
    h\left(\hat{\Pi}(\nu^{a,d}), y_\pi\right) = \langle  \hat{\nu}^{a,d}_\pi , y_\pi \rangle,
\end{equation}
Using the Minkowski additive property of support functions from \eqref{eq:support_function_minkowski_additive} we can write 
\begin{align*}
    h\left(\biguplus_{a,d} \hat{\Pi}(\nu^{a,d}), y_\pi\right) 
        &= \sum_{a,d} h(\hat{\Pi}(\nu^{a,d}), y_\pi)\\
        &=\sum_{a,d} \left\langle  \hat{\nu}^{a,d}_\pi, y_\pi \right\rangle\\
        &=\left\langle \sum_{a,d}  \hat{\nu}^{a,d}_\pi, y_\pi \right\rangle\\
\end{align*}
and so $\mu_\pi = \sum_{a,d} \hat{\nu}^{a,d}_\pi$ is a vertex of $\biguplus_{a,d} \hat{\Pi}(\nu^{a,d})$.
\end{proof}

\subsection{Disaggregation proof}\label{appendix:disaggregation_proof}
\begin{proof}
As stated in Sec. \ref{section:reformulation}: $\hat{A}^{a,d}\hat{v}^i \in \hat{\Pi}(v^i) \; \forall \hat{A}^{a,d} \in \hat{B}^{a,d}$,
from Theorem \ref{theorem:majorization_permutahedron}, $\hat{\Pi}(v^i) \equiv \mathcal{P}^i$ therefore \eqref{eq:disaggregation_criteria_feasible} holds. We can verify this mapping preserves the optimal aggregate charging profile:
\begin{align*}
    \sum_i^K u^i &= \sum_{a,d} \sum_{i\in \mathbb{A}^{a,d}} \hat{A}^{a,d^*} \hat{v}^i \\
                 &= \sum_{a,d}  \hat{A}^{a,d^*} \hat{\nu}^{a,d} = x^*,\\
\end{align*}
where we have used the definition of $\nu^{a,d}$ from \eqref{eq:major_vertex_sum} in the last line, and so \eqref{eq:disaggregation_criteria_tracks} is satisfied.
\end{proof}



\end{document}

%% file: figures/cubePlane.tex
\begin{tikzpicture}[scale=3]
\coordinate (A) at (0,0,0);
\coordinate (B) at (1,0,0);
\coordinate (C) at (1,1,0);
\coordinate (D) at (0,1,0);
\coordinate (E) at (0,0,1);
\coordinate (F) at (1,0,1);
\coordinate (G) at (1,1,1);
\coordinate (H) at (0,1,1);

\draw[->] (0,0,0) -- (1.2,0,0) node[below right]{$u_1$};
\draw[->] (0,0,0) -- (0,1.2,0) node[right]{$u_2$};
\draw[->] (0,0,0) -- (0,0,1.2) node[above]{$u_3$};

\draw (A) -- (B) -- (C) -- (D) -- cycle;
\draw (E) -- (F) -- (G) -- (H) -- cycle;
\draw (A) -- (E);
\draw (B) -- (F);
\draw (C) -- (G);
\draw (D) -- (H);

\filldraw[opacity=0.5, blue] (1,0.7,0) -- (1,0,0.7) -- (0.7,0,1) -- (0,0.7,1)--  (0,1,0.7)-- (0.7,1,0)  -- cycle;

\coordinate (I) at (1,0.7,0);
\coordinate (J) at (0,0.7,1);
\coordinate (K) at (0.7,1,0);
\coordinate (L) at (0.7,0,1);
\coordinate (M) at (1,0,0.7);
\coordinate (N) at (0,1,0.7);

\filldraw (I) circle (0.5pt) node[above right] {\footnotesize $(m,r,0)$};
\filldraw (J) circle (0.5pt) node[above left] {\footnotesize $(0,r,m)$};
\filldraw (K) circle (0.5pt) node[above right] {\footnotesize $(r,m,0)$};
\filldraw (L) circle (0.5pt) node[below left] {\footnotesize $(r,0,m)$};
\filldraw (M) circle (0.5pt) node[below right] {\footnotesize $(m,0,r)$};
\filldraw (N) circle (0.5pt) node[above left] {\footnotesize $(0,m,r)$};

\end{tikzpicture}

%% file: root.bbl
\begin{thebibliography}{10}
\providecommand{\url}[1]{#1}
\csname url@samestyle\endcsname
\providecommand{\newblock}{\relax}
\providecommand{\bibinfo}[2]{#2}
\providecommand{\BIBentrySTDinterwordspacing}{\spaceskip=0pt\relax}
\providecommand{\BIBentryALTinterwordstretchfactor}{4}
\providecommand{\BIBentryALTinterwordspacing}{\spaceskip=\fontdimen2\font plus
\BIBentryALTinterwordstretchfactor\fontdimen3\font minus
  \fontdimen4\font\relax}
\providecommand{\BIBforeignlanguage}[2]{{%
\expandafter\ifx\csname l@#1\endcsname\relax
\typeout{** WARNING: IEEEtran.bst: No hyphenation pattern has been}%
\typeout{** loaded for the language `#1'. Using the pattern for}%
\typeout{** the default language instead.}%
\else
\language=\csname l@#1\endcsname
\fi
#2}}
\providecommand{\BIBdecl}{\relax}
\BIBdecl

\bibitem{Lee2019ACN-Data:Dataset}
Z.~J. Lee, T.~Li, and S.~H. Low, ``{ACN-Data: Analysis and Applications of an
  Open EV Charging Dataset},'' in \emph{Proceedings of the Tenth ACM
  International Conference on Future Energy Systems}, ser. e-Energy '19.\hskip
  1em plus 0.5em minus 0.4em\relax New York, NY, USA: Association for Computing
  Machinery, 2019, pp. 139--149.

\bibitem{Callaway2011AchievingLoads}
D.~S. Callaway and I.~A. Hiskens, ``{Achieving Controllability of Electric
  Loads},'' \emph{Proceedings of the IEEE}, vol.~99, no.~1, pp. 184--199, 2011.

\bibitem{EsmaeilZadehSoudjani2015AggregationAbstractions}
S.~Esmaeil Zadeh~Soudjani and A.~Abate, ``{Aggregation and Control of
  Populations of Thermostatically Controlled Loads by Formal Abstractions},''
  \emph{IEEE Transactions on Control Systems Technology}, vol.~23, no.~3, pp.
  975--990, 2015.

\bibitem{Hao2015AggregateLoads}
H.~Hao, B.~M. Sanandaji, K.~Poolla, and T.~L. Vincent, ``{Aggregate Flexibility
  of Thermostatically Controlled Loads},'' \emph{IEEE Transactions on Power
  Systems}, vol.~30, no.~1, pp. 189--198, 2015.

\bibitem{Hao2014CharacterizingLoads}
H.~Hao and W.~Chen, ``{Characterizing flexibility of an aggregation of
  deferrable loads},'' in \emph{53rd IEEE Conference on Decision and Control},
  2014, pp. 4059--4064.

\bibitem{Barot2017APolytopes}
S.~Barot and J.~A. Taylor, ``{A concise, approximate representation of a
  collection of loads described by polytopes},'' \emph{International Journal of
  Electrical Power {\&} Energy Systems}, vol.~84, pp. 55--63, 2017.

\bibitem{Trangbaek2011ExactControl}
K.~Trangbaek, M.~P{\'{e}}tersen, J.~Bendtsen, and J.~Stoustrup, ``{Exact power
  constraints in smart grid control},'' \emph{Proceedings of the IEEE
  Conference on Decision and Control}, pp. 6907--6912, 2011.

\bibitem{Zhao2017ALoads}
L.~Zhao, W.~Zhang, H.~Hao, and K.~Kalsi, ``{A Geometric Approach to Aggregate
  Flexibility Modeling of Thermostatically Controlled Loads},'' \emph{IEEE
  Transactions on Power Systems}, vol.~32, no.~6, pp. 4721--4731, 2017.

\bibitem{Muller2019AggregationResources}
F.~L. M{\"{u}}ller, J.~Szab{\'{o}}, O.~Sundstr{\"{o}}m, and J.~Lygeros,
  ``{Aggregation and disaggregation of energetic flexibility from distributed
  energy resources},'' \emph{IEEE Transactions on Smart Grid}, vol.~10, no.~2,
  pp. 1205--1214, 3 2019.

\bibitem{Nazir2018InnerResources}
M.~S. Nazir, I.~A. Hiskens, A.~Bernstein, and E.~Dall'Anese, ``{Inner
  Approximation of Minkowski Sums: A Union-Based Approach and Applications to
  Aggregated Energy Resources},'' in \emph{2018 IEEE Conference on Decision and
  Control (CDC)}, 2018, pp. 5708--5715.

\bibitem{Kundu2018ApproximatingApproach}
S.~Kundu, K.~Kalsi, and S.~Backhaus, ``{Approximating Flexibility in
  Distributed Energy Resources: A Geometric Approach},'' in \emph{2018 Power
  Systems Computation Conference (PSCC)}, 2018, pp. 1--7.

\bibitem{Weibel2007MinkowskiComputation}
C.~Weibel, ``{Minkowski sums of polytopes: combinatorics and computation},''
  Ph.D. dissertation, Citeseer, 2007.

\bibitem{Zhao2016ExtractingApproximation}
L.~Zhao, H.~Hao, and W.~Zhang, ``{Extracting flexibility of heterogeneous
  deferrable loads via polytopic projection approximation},'' in \emph{2016
  IEEE 55th Conference on Decision and Control (CDC)}, 2016, pp. 6651--6656.

\bibitem{Ozturk2022AggregationAlgorithms}
E.~{\"{O}}zt{\"{u}}rk, K.~Rheinberger, T.~Faulwasser, K.~Worthmann, and
  M.~Prei{\ss}inger, ``{Aggregation of Demand-Side Flexibilities: A Comparative
  Study of Approximation Algorithms},'' \emph{Energies 2022, Vol. 15, Page
  2501}, vol.~15, no.~7, p. 2501, 3 2022.

\bibitem{Kaibel2011ExtendedOptimization}
V.~Kaibel, ``{Extended formulations in combinatorial optimization},''
  \emph{arXiv preprint arXiv:1104.1023}, 2011.

\bibitem{Tiwary2008OnPolytopes}
H.~R. Tiwary, ``{On the Hardness of Computing Intersection, Union
  and Minkowski Sum of Polytopes},'' \emph{Discrete {\&} Computational
  Geometry}, vol.~40, no.~3, pp. 469--479, 2008.

\bibitem{Dahl2010Majorization01-matrices}
G.~Dahl, ``{Majorization permutahedra and (0,1)-matrices},'' \emph{Linear
  Algebra and its Applications}, vol. 432, no.~12, pp. 3265--3271, 2010.

\bibitem{Postnikov2009PermutohedraBeyond}
A.~Postnikov, ``{Permutohedra, Associahedra, and Beyond},'' \emph{International
  Mathematics Research Notices}, vol. 2009, no.~6, pp. 1026--1106, 1 2009.

\bibitem{Ziegler2012LecturesPolytopes}
G.~M. Ziegler, \emph{{Lectures on polytopes}}.\hskip 1em plus 0.5em minus
  0.4em\relax Springer Science {\&} Business Media, 2012, vol. 152.

\bibitem{Fukuda2004FromPolytopes}
K.~Fukuda, ``{From the zonotope construction to the Minkowski addition of
  convex polytopes},'' \emph{Journal of Symbolic Computation}, vol.~38, no.~4,
  pp. 1261--1272, 2004.

\bibitem{Papadaskalopoulos2013DecentralizedMechanism}
D.~Papadaskalopoulos and G.~Strbac, ``{Decentralized Participation of Flexible
  Demand in Electricity Markets—Part I: Market Mechanism},'' \emph{IEEE
  Transactions on Power Systems}, vol.~28, no.~4, pp. 3658--3666, 2013.

\bibitem{Gharesifard2016Price-basedResources}
B.~Gharesifard, T.~Ba{\c{s}}ar, and A.~D. Dom{\'{i}}nguez-Garc{\'{i}}a,
  ``{Price-based coordinated aggregation of networked distributed energy
  resources},'' \emph{IEEE Transactions on Automatic Control}, vol.~61, no.~10,
  pp. 2936--2946, 2016.

\end{thebibliography}
